\documentclass[conference]{IEEEtran}
\IEEEoverridecommandlockouts
\usepackage{cite}
\usepackage{amsmath,amssymb,amsfonts}
\usepackage{algorithmic}
\usepackage{bm}
\usepackage{graphicx}
\usepackage{dsfont}
\usepackage{enumitem}
\usepackage{pifont}
\usepackage{siunitx}
\usepackage{textcomp}
\usepackage{xcolor}
\usepackage{amsthm}
\usepackage{mathrsfs}  % mathscr
\def\BibTeX{{\rm B\kern-.05em{\sc i\kern-.025em b}\kern-.08em
    T\kern-.1667em\lower.7ex\hbox{E}\kern-.125emX}}

\definecolor{emerald}{rgb}{0,0.65,0.42}

\definecolor{blue}{rgb}{0,0,0.8}

\definecolor{blue}{rgb}{0,0,0.8}

\allowdisplaybreaks[3]
\makeatletter
\newcommand{\vast}{\bBigg@{5.5}}
\newcommand{\Vast}{\bBigg@{6}}

\newtheorem{proposition}{Proposition}

\begin{document}

\title{Distributed Frequency Regulation for Heterogeneous Microgrids via Steady State Optimal Control\\
\thanks{This work was funded by the Deutsche Forschungsgemeinschaft (DFG, German Research Foundation)---project number 360464149.	
}
}

\author{\IEEEauthorblockN{Lukas K\"olsch, Manuel Dupuis, Kirtan Bhatt, Stefan Krebs, and S\"oren Hohmann}
	\IEEEauthorblockA{\textit{Institute of Control Systems, Karlsruhe Institute of Technology (KIT)}, Karlsruhe, Germany \\
		lukas.koelsch@kit.edu, manuel.dupuis@student.kit.edu, kirtan.bhatt@student.kit.edu, \\ stefan.krebs@kit.edu, soeren.hohmann@kit.edu}
}

\maketitle

\begin{abstract}
In this paper, we present a model-based frequency controller for microgrids with nonzero line resistances 
based on a port-Hamiltonian formulation of the microgrid model and real-time dynamic pricing. 
The controller is applicable for conventional generation with synchronous machines as well as for power electronics interfaced sources and it is robust against power fluctuations from uncontrollable loads or volatile regenerative sources.
The price-based formulation allows additional requirements such as active power sharing to be met.
The capability and effectiveness of our procedure is demonstrated by means of an 18-node exemplary grid.
\end{abstract}


\section{Introduction}
\subsection{State of Research}
The energy transition motivates a worldwide trend towards renewable energy generation which should substitute the conventional power plants in the future. 
A key aspect of renewable sources is their distributed and volatile nature compared to the centralized and well predictable character of conventional power plants \cite{Kariniotakis_2017,Bevrani.2013b}. 
This also results in a necessary change in the control schemes of the power network. 
So far, frequency control, i.e. the regulation of the imbalance between power generation and demand, has been the task of the transmission system operator using a hierarchy of primary, secondary and tertiary frequency control layers: 
In the first layer, frequency deviations and thus power imbalances are prevented from further increasing, in the second layer, the nominal state is restored, and in the third layer, an economic optimization is carried out. Both secondary and tertiary layer are each governed by a central controller.
%hier ist noch ein Bruch drin, denn was soll das mit zentral zu tun haben?

However, the distributed nature of renewable energy generation encourages the application of a distributed frequency control scheme between multiple agents which are able to handle the control task in parallel 
\cite{Patnaik.2018}.
%At present, several publications propose distributed control strategies based on principles of the current methods, e.g. in \cite{Dorfler.2016} the well-established concept of droop control is applied to microgrids and therefore gives a realization of a distributed control strategy based on well-established methods.
For this reason, steady state optimal control by real-time dynamic pricing poses an advantageous control concept especially for large scale networks, since it enables communication of network imbalances via a price signal, see \cite{Doerfler2019} for a survey on current research directions regarding frequency regulation. This kind of controller features a distributed architecture for frequency restoration based on neighbor-to-neighbor communication and local measurements which is able to reach a desired economic optimum at steady state and thus provides a unifying approach incorporating all three control layers \cite{Koelsch2019}. \\
A common assumption made in previous publications on dynamic pricing methods for frequency regulation, e.g. \cite{Stegink.}, is that all power lines are lossless, which is in fact an incorrect assumption especially for medium and low voltage grids. Practically, if controllers like the one proposed in \cite{Stegink.} are used, a synchronous frequency $\overline{\omega}$ is achieved which deviates from the nominal frequency $\omega^{n}$ \cite{Koelsch2019}.

There are publications like \cite{Bevrani.2013b} which take nonzero line resistances into account by proposing a generalized droop control method. However, these methods rely on a fixed $R/X$ ratio for the designed generalized droop control concept, which makes it necessary to compensate for this simplifying assumption by an additional control layer.      

\subsection{Main Contributions}
A \emph{gradual} transition from today's power network with large, centralized power plants towards a future network with a large number of small, decentralized generation units is essential. 
For this purpose, we present a model-based, distributed, steady state optimal frequency controller that is applicable for heterogeneous and lossy microgrids with both types of network connectors, i.e. synchronous generators for conventional power plants as well as inverters for renewable sources. 
The controller design is based on \cite{Stegink.2017,Stegink.} and our previous work \cite{Koelsch2019} by  integrating an inverter model based on \cite{Monshizadeh.} in the underlying microgrid model. Both system and controller are represented as a port-Hamiltonian system, which results in a closed-loop system that is again port-Hamiltonian.
Due to the port-Hamiltonian structure, stability of the closed-loop system can finally be characterized using a (shifted) passivity property.
%Stability of the closed-loop system is derived and the system behavior is shown in various simulations. 

The remainder of this paper is structured as follows. In section \ref{section:model}, we derive a port-Hamiltonian model of a heterogeneous microgrid consisting of a mixture of conventional generation with synchronous machines, renewable generation via power electronics interfaced sources, and uncontrollable consumers or producers. In section \ref{section:controller}, we formulate a price-based, distributed frequency controller which is robust against power demand fluctuations.
In section \ref{section:simulation}, we demonstrate the performance of the controller under heavy load changes by means of an 18-node test network and in section \ref{section:conclusion}, we summarize our results and provide an outlook on future research directions.

\section{Microgrid Model}\label{section:model}
Microgrids consist of conventional and regenerative generators as well as consumers that are all physically connected together via a lossy and meshed electrical network. 

Accordingly, the microgrid is modeled as a directed graph consisting of three different types of nodes:
\begin{enumerate}
\item \emph{Synchronous generator nodes} which are connected to synchronous generators of conventional power plants.
\item \emph{Inverter nodes} which are connected to power electronics interfaced sources.
\item \emph{Load nodes} which are characterized by a given and uncontrollable active and reactive power demand. 
\end{enumerate}
%In this process, uncontrollable generation can be modeled by load nodes with negative active power demands.

To set up a dynamic model of the microgrid in port-Hamiltonian formulation in the course of this section, at first all notational conventions as well as symbols used are outlined (section \ref{ch:Notational}) and the model assumptions and simplifications used are listed (section \ref{ch:Assumptions}), before submodels for synchronous generators (section \ref{ch:g-Model}), inverters (section \ref{ch:i-model}) and load nodes (section \ref{ch:l-model}) are derived. Interconnection with lossy lines (section \ref{ch:line-model}) finally results in an overall model in port-Hamiltonian representation (section \ref{ch:overall-model}), which forms the basis (``plant model'') for the controller design.

\subsection{Notational Preliminaries}\label{ch:Notational}
Vector $\bm a = \mathrm{col}_i\{a_i\}=\mathrm{col}\{a_1,a_2,\ldots\}$ is a column vector of elements $a_i$, $i=1,2,\ldots$ and matrix $\bm A=\mathrm{diag}_i\{a_i\}=\mathrm{diag}\{a_1,a_2,\ldots\}$ is a (block-)diagonal matrix of elements $a_i$, $i =1,2,\ldots$. The $(n \times n)$-identity matrix and $(n \times n)$-zero matrix are denoted by $\bm I_n$ and $\bm 0_n$, respectively. Steady state (i.e. equilibrium) variables are marked with an overline.
%For all other vectors and matrices, the dimensions are either explicitly specified or they result from the context.

The microgrid is modeled by a directed graph $\mathscr G_p=(\mathcal V, \mathcal E_p)$ with $\mathcal V = \mathcal V_\mathcal G \cup \mathcal V_I \cup \mathcal V_\ell$ being the set of $n_\mathcal G=|\mathcal V_\mathcal G|$ generator nodes , $n_\mathcal I =|\mathcal V_I|$ inverter nodes, and $n_\ell=|\mathcal V_\ell|$ load nodes.

The physical interconnection of the nodes is represented by the incidence matrix $\bm D_p \in \mathds R^{n\times m_p}$ with $n=n_\mathcal G+n_\mathcal I+n_\ell$ and $m_p = |\mathcal E_p|$. Incidence matrix $\bm D_p$ can be subdivided as follows
\begin{align}
\bm D_p = \begin{bmatrix} \bm D_{p\mathcal G} \\ \bm D_{p\mathcal I} \\ \bm D_{p \ell} \end{bmatrix},
\end{align}
where submatrices $\bm D_{p \mathcal G}$, $\bm D_{p\mathcal I}$ and $\bm D_{p \ell}$ correspond to the generator, inverter, and load nodes, respectively.
We note $j \in \mathcal N_i$ if node $j$ is a \emph{neighbor} of node $i$, i.e. $j$ is adjacent to $i$ in the undirected graph.

Positive semidefiniteness of a matrix is denoted by $\succeq 0$, whereas nonnegativity of a scalar is denoted by $\geq 0$.

Subscript $p$ denotes the \emph{plant} variables, i.e. the variables of the microgrid model, whereas subscript $c$ denotes the variables of the \emph{controller}.

A list of symbols used for parameters and state variables of the microgrid model is given in Table \ref{table:PlantSymbols}.
\begin{table}[!t]
	\renewcommand{\arraystretch}{1.3}
	\caption{List of Microgrid Parameters and State Variables}
	\begin{center}
		\begin{tabular}{|l||l|}
			\hline
			$A_i$ & positive generator, inverter and load damping constant \\
			$B_{ii}$ & negative of self-susceptance \\
			$B_{ij}$ & negative of susceptance of line $(i,j)$ \\
			$C_{DC}$ & capacitance in DC circuit of the inverter \\
			$\bm D_p$ & incidence matrix of microgrid\\
			$G_{DC}$ & conductance in DC circuit of the inverter \\
			$G_{ij}$ & negative of conductance of line $(i,j)$\\
			$i_{DC}$ & current source in DC circuit of inverter \\
			$\bm i_{\alpha \beta}$ & output current of inverter \\
			$L_i$ & deviation of angular momentum from nominal value $M_i\omega^{n}$ \\
			$M_i$ & moment of inertia \\ 
			$p_i$ & sending-end active power flow \\
			$p_{g,i}$ & active power generation \\
			$p_{\ell,i}$ & active power demand \\
			$q_i$ & sending-end reactive power flow \\
			$q_{\ell,i}$ & reactive power demand \\
			$T_{e}$ & electrical torque at generator \\
			$T_{m}$ & mechanical torque at generator \\
			$u_{DC}$ & input voltage of the inverter \\
			$U_i$ & magnitude of transient internal voltage \\
			$U_{f,i}$ & magnitude of excitation voltage \\
			$X_{d,i}$ & d-axis synchronous reactance \\
			$X'_{d,i}$ & d-axis transient reactance \\
			$\theta_i$ & bus voltage phase angle \\
			$\vartheta_{ij}$ & bus voltage angle difference $\theta_i-\theta_j$\\
			$\Phi$ & overall transmission losses \\
			$\tau_{U,i}$ & open-circuit transient time constant of synchronous machine \\
			$\omega_i$ & deviation of bus frequency from nominal value $\omega^{n}$ \\
			$\omega_{\mathcal I}$ & virtual frequency of inverter \\
			\hline
		\end{tabular}
		\label{table:PlantSymbols}
	\end{center}
\end{table}
\subsection{Modeling Assumptions}\label{ch:Assumptions}
In accordance with \cite{Stegink.2017,Stegink.,Monshizadeh.}, we make the following modeling assumptions for the microgrid model and the controller:
%\begin{enumerate}[label=\protect\colorbox{yellow}{protect\ding{\value},start=192]
\begin{enumerate}[leftmargin=2.5em,label=(A\arabic*)]
	\item The grid is operating around the nominal frequency $\omega^{n}=2\pi \cdot \SI{50}{\hertz}$. \label{A1}
	\item The grid is a balanced three-phased system and the lines are represented by its one-phase $\pi$-equivalent circuits. \label{A2}
	\item Subtransient dynamics of the synchronous generators is neglected. \label{A3}
	\item The matching controller of the inverters presented in section \ref{ch:i-model} has fast dynamics compared to the price-based frequency controller. \label{A4}
\end{enumerate}
However, we make the following less restrictive assumptions:
\begin{enumerate}[leftmargin=2.5em,label=(A\arabic*),resume]
	\item Power lines are lossy, i.e. have nonzero resistances. \label{A5}
	\item Loads do not have to be constant. \label{A6}
	\item Excitation voltages of the generators do not have to be constant. \label{A7}
\end{enumerate}

\subsection{Dynamic Model of Generator Nodes}\label{ch:g-Model}

For generator node $i \in \mathcal V_\mathcal G$, the third-order generator model ("flux-decay" model), described in local dq-frame, appropriately represents the transient dynamic behavior \cite{Machowski.2012}:
\begin{IEEEeqnarray}{rCl} 
\dot{\theta}_i &=& \omega_{i}, \label{G1}
\\
M_i\dot{\omega}_i &=& -A_i\omega_i + p_{g,i} - p_{\ell,i} - p_{i}, \label{G2}
\\
\tau_{U,i}\dot{U}_i &=& U_{f,i} - U_i + \left(X_{d,i}-X_{d,i}^{'}\right) I_{d,i}. \label{G3}
\end{IEEEeqnarray}
According to \cite{Boldea.2006}, the stator d-axis current $I_{d,i}$ can be formulated as
\begin{equation}
I_{d,i} = \frac{U_{j}\cos\left(\theta_{i}-\theta_{j}\right)-U_{i}}{X_{d,i}} \label{G4}
\end{equation}
with transient internal voltage $U_i$ and terminal voltage $U_j$. In power system literature, it is a common assumption that the stator resistance can be neglected and thus lossless reactive power flow
\begin{equation}
q_i = \frac{U_i^{2}}{X_d} - \frac{U_iU_j}{X_d}\cos\left(\theta_i-\theta_j\right) \label{G5}
\end{equation}
can be used to describe the generator dynamics. With \eqref{G5}, the identity
\begin{IEEEeqnarray}{rCl}
\frac{q_i}{U_i} &=& - \frac{U_i}{X_{d,i}}  + \frac{U_j}{X_{d,i}}\cos\left(\theta_i-\theta_j\right) \nonumber 
\\
&=& I_{d,i} \label{G7}	
\end{IEEEeqnarray}
holds. Substituting \eqref{G3} with \eqref{G7}, this yields the generator model as in \cite{Stegink.2017}:
\begin{align} 
	\dot{\theta}_i &= \omega_{i}, && i \in \mathcal V_\mathcal G, \label{G8}
	\\
	M_i\dot{\omega}_i &= -A_i\omega_i + p_{g,i} - p_{\ell,i} - p_{i}, && i \in \mathcal V_\mathcal G, \label{G9}
	\\
	\tau_{U,i}\dot{U}_i &= U_{f,i} - U_i + \left(X_{d,i}-X_{d,i}^{'}\right)U_{i}^{-1}q_i, && i \in \mathcal V_\mathcal G. \label{G10}
\end{align}
Without loss of generality, we assume that the generated power $p_{g,i}$ is controllable, while the power demand $p_{\ell,i}$ is uncontrollable. Thus $p_{g,i}$ serves as control input and $p_{\ell,i}$ is a disturbance input. 
%If there are controllable loads, they are modeled as negative generation. 	  

\subsection{Dynamic Model of Inverter Nodes}\label{ch:i-model}
The inverters are regulated by $v_{DC} \sim  \omega$ matching control \cite{Monshizadeh.} to mimic the dynamic behavior of a synchronous generator. 
%Therefore an inverter, described in ($\alpha \beta$)-frame, with a DC-side capacitor is used, which enables a formal description that matches the above third-order generator model in terms of the swing equation \eqref{G9}, 
This is done by exploiting the structural similarities between kinetic energy of the rotor of a synchronous generator and electric energy stored in the DC-side capacitor of an inverter.
%where the kinetic energy of the rotor is substituted by the stored energy in the DC-side capacitor.
The relevant equation describing the used 3-phase DC/AC inverter dynamics in $\alpha \beta$-frame is 
\begin{equation}
C_{DC}\dot{u}_{DC} = - G_{DC}u_{DC} + i_{DC} - \frac{1}{2}\bm i_{\alpha \beta }^{\top}\bm m_{\alpha \beta} , \label{I1}
\end{equation}
with $\bm m_{\alpha \beta}$ being an AC power electronics modulation signal generated by the controller to match the behavior of the synchronous generator presented above \cite{Jouini.2016}. For this purpose, the internal model, also in $ \alpha \beta $-frame, equals \cite{Jouini.2016}
\begin{IEEEeqnarray}{rCl}
\dot{\theta} &=& \omega, \label{I2}
\\
M\dot{\omega} &=& -A \omega + T_m - T_e, \label{I3}
\end{IEEEeqnarray}
where the expression for the electrical torque $T_e$ can be expressed by 
\begin{equation}
T_e = - L_{m}i_{f}\bm i_{\alpha \beta}^{\top}\left[ \begin{array}{c} -\sin\left(\theta\right) \\ \cos\left(\theta\right) \end{array} \right] , \label{I4}
\end{equation}  
with the stator-to-rotor mutual inductance $L_m$ and the excitation current $i_f$.
Hence, the internal generator model to be matched by the inverter is
\begin{IEEEeqnarray}{rCL}
\dot{\theta} &=& \omega, \label{I5}
\\
M\dot{\omega} &=& -A \omega + T_m + L_{m}i_{f}\bm i_{\alpha \beta}^{\top}\left[ \begin{array}{c} -\sin\left(\theta\right) \\ \cos\left(\theta\right) \end{array} \right]. \label{I6}
\end{IEEEeqnarray} 
For this purpose, the modulation signal is chosen to  
\begin{equation}
\bm m_{\alpha \beta} = \mu \cdot  \left[ \begin{array}{c} -\sin\left(\theta_{\mathcal I}\right) \\ \cos\left(\theta_{\mathcal I}\right) \end{array} \right] , \label{I7} 
\end{equation}
with $\theta_{\mathcal I}$ as a ``virtual'' rotor angle and a constant gain $\mu >0$. The ``virtual'' frequency thus results in $\omega_{\mathcal I} = \dot{\theta}_{\mathcal I}$. 

Furthermore, a linking between rotational speed, or frequency respectively, and power consumption needs to be created by construction. This is achieved by appyling a proportional controller for $\omega_{\mathcal I}$ based on local measurement of $u_{DC}$:
\begin{equation}
\omega_{\mathcal I} = \eta u_{DC} , \label{I8} 
\end{equation}
again with a constant gain $\eta >0$. \\
By setting $\mu$ to 
\begin{equation}
\mu = - 2 \eta L_m i_f \label{I9}
\end{equation}
and using \eqref{I7}, the last term in \eqref{I1} can be reformulated to
\begin{IEEEeqnarray}{rCL}
\frac{1}{2}\bm i_{\alpha \beta}^{\top}\bm m_{\alpha \beta} &=& - \eta  L_m i_f \bm i_{\alpha \beta}^{\top}\left[ \begin{array}{c} -\sin\left(\theta_{\mathcal I}\right) \\ \cos\left(\theta_{\mathcal I}\right) \end{array} \right] \nonumber
\\
&=& \eta T_{e,\mathcal I} , \label{I11}
\end{IEEEeqnarray}
with the "virtual" electrical torque $T_{e,\mathcal I}$. \\
By inserting \eqref{I11} in \eqref{I1}, substituting $u_{DC}$ w.r.t \eqref{I8} and dividing by $\eta$, the inverter model can be formulated as
\begin{IEEEeqnarray}{rCL}
\dot{\theta}_{\mathcal I} &=& \omega_{\mathcal I}, \label{I12}
\\
\frac{C_{DC}}{\eta^{2}}\dot{\omega}_{\mathcal I} &=& - \frac{G_{DC}}{\eta^{2}}\omega_{\mathcal I} + \frac{i_{DC}}{\eta} - T_{e,\mathcal I} . \label{I13}
\end{IEEEeqnarray}
To further highlight the resemblance between the modulated inverter equations and the synchronous generator model, as in \cite{Jouini.2016}, the coefficients in \eqref{I13} can be interpreted as \emph{virtual inertia} $M^{*}_\mathcal I = \frac{C_{DC}}{\eta^{2}}$ and \emph{virtual damping} $A^{*}_\mathcal I = \frac{G_{DC}}{\eta^{2}}$. Moreover, the DC current $i_{DC}$ is chosen according to \cite{Monshizadeh.} as
\begin{align}
i_{DC} = \eta A^{*}_\mathcal I \omega^{n} + \eta \cdot \frac{p_g}{\omega_{\mathcal I}}. \label{I13b}
\end{align}
Furthermore, the \emph{virtual electrical torque} can be described as $T_{e,\mathcal I} = \frac{p_{\mathcal I}}{\omega_{\mathcal I}}$, with $p_{\mathcal I}$ being the power input of the inverter. Under the assumption that no power is dissipated in the inverter $i \in \mathcal V_\mathcal I$, i.e. $p_{\mathcal I,i}=p_{\ell,i}+p_i$, and use of \eqref{I8} and \eqref{I13b}, this allows a reformulation of \eqref{I13} as
\begin{align}
	M^{*}_{\mathcal I,i} \dot{\omega}_{\mathcal I,i} = - A^{*}_{\mathcal I,i}\left(\omega_{\mathcal I,i} - \omega^{n} \right) + \frac{1}{\omega_{\mathcal I}} \left(p_{g,i} - p_{\ell,i} - p_i\right) . \label{I13c}
\end{align}
If \ref{A1} holds, multiplying \eqref{I13c} with $\omega^{n}$ yields 
\begin{align}
	\dot{\theta}_{\mathcal I,i} &= \omega_{\mathcal I,i}, && i \in \mathcal V_\mathcal I, \label{I14}
	\\
	M_{\mathcal I,i}\dot{\omega}_{\mathcal I,i} &= - A_{\mathcal I,i}\omega_{\mathcal I,i} + p_{g,i} - p_{\ell,i} , && i \in \mathcal V_\mathcal I \label{I15}
\end{align}
with $M_{\mathcal I} = M^{*}_{\mathcal I} \omega^{n}$, $A_{\mathcal I} = A^{*}_{\mathcal I} \omega^{n}$, and $\omega_{\mathcal I}$ expressing the deviation of frequency from its nominal value $\omega^{n}$. In particular, the structural similarities between  \eqref{I15} and the swing equation \eqref{G9} can be clearly seen now. 
%After we have completed the modeling of the inverter nodes and the required matching control, it should be mentioned at this point, that Assumption \ref{A4} holds. 
\subsection{Dynamic Model of Load Nodes}\label{ch:l-model}
The loads are modeled by an active power consumption which consists of both a frequency-dependent part with load damping coefficients $A_{i}\geq 0$ and a frequency-independent part $p_\ell$, as well as frequency-independent reactive consumption $q_\ell$ \cite{Boldea.2006}:
\begin{align}
\dot \theta_i &= \omega_i, && i \in \mathcal V_\ell, \label{load1}\\
0 &= -A_i \omega_i -p_{\ell,i}-p_i, && i \in \mathcal V_\ell,  \label{load2} \\
0 &= -q_{\ell,i}-q_i, && i \in \mathcal V_\ell.  \label{load3}
\end{align}
\subsection{Transmission Lines}\label{ch:line-model}
Generator, inverter, and load nodes are interconnected via transmission lines, which are modeled by the lossy AC power flow equations \cite{Machowski.2012}
\begin{align}
p_i &= \sum_{j \in \mathcal N_i} B_{ij}U_iU_j \sin(\vartheta_{ij}) + G_{ii}U_i^2 && \nonumber\\
&+ \sum_{j \in \mathcal N_i}G_{ij}U_iU_j\cos(\vartheta_{ij}),&& i \in \mathcal V, \label{powerflow1}\\
q_i &= -\sum_{j \in \mathcal N_i} B_{ij}U_iU_j \cos(\vartheta_{ij}) + B_{ii}U_i^2 && \nonumber\\
&+ \sum_{j \in \mathcal N_i}G_{ij}U_iU_j\sin(\vartheta_{ij}),&& i \in \mathcal V \label{powerflow2}
\end{align}
with $\bm Y = \bm G + \mathrm j \bm B$ being the admittance matrix and $\vartheta_{ij}=\theta_i - \theta_j$ being the voltage angle deviation between two adjacent nodes. 
%\lukas{Satz evtl. weglassen:} 
Note that by definition of the admittance matrix, $G_{ij}<0$ and $B_{ij}>0$ if nodes $i$ and $j$ are connected via a resistive-inductive line \cite{Machowski.2012}.

\subsection{Overall Model}\label{ch:overall-model}
The equations for generator \eqref{G8}--\eqref{G10}, inverter \eqref{I14}--\eqref{I15} and load nodes \eqref{load1}--\eqref{load3} can be summarized in a compact notation as follows: 
\begin{align}
\dot \theta_i &= \omega_i, && i \in \mathcal V, \label{planteq1}\\
\dot L_i &= -A_i \omega_i + p_{g,i}-p_{\ell,i}-p_i, && i \in \mathcal V_\mathcal G \cup \mathcal V_\mathcal I, \label{planteq2} \\
\tau_{d,i}\dot U_i &= U_{f,i}-U_i - \frac{X_{d,i}-X'_{d,i}}{U_i} \cdot q_i, && i \in \mathcal V_\mathcal G,  \label{planteq3}\\
0 &= -A_i \omega_i -p_{\ell,i}-p_i, && i \in \mathcal V_\ell, \label{planteq4}\\
0 &= -q_{\ell,i}-q_i, && i \in \mathcal V_\ell. \label{planteq5}
\end{align}
The interconnection of these node dynamics with the power flow equations \eqref{powerflow1}--\eqref{powerflow2} leads to the overall model of the microgrid, which is presented in port-Hamiltonian form.

First, the ``plant'' state vector $\bm x_p$ of the microgrid is defined as
\begin{align}
\bm x_p = \mathrm{col}\{\bm\vartheta,\bm L_\mathcal G,\bm L_\mathcal I,\bm U_g,\bm\omega_\ell,\bm U_\ell\} \label{plant-states}
\end{align}
with voltage angle deviations $\bm \vartheta = \bm D_p \bm \theta$ and $\bm L_\mathcal G, \bm L_\mathcal I$ being the vectors of angular momentum deviations $L_i = M_i \cdot \omega_i$ of generator and inverter nodes, respectively.
The state vector is used to set up the plant Hamiltonian
\begin{align}
H_p(\bm x_p) &= \frac 12 \sum_{i \in \mathcal V_\mathcal G}\left( M_i^{-1}L_i^2 + \frac{U_i^2}{X_{d,i}-X'_{d,i}}\right) \nonumber \\
& + \frac 12 \sum_{i \in \mathcal V_\mathcal I} M_i^{-1}L_i^2 \nonumber \\
&- \frac 12 \sum_{i \in \mathcal V} B_{ii}U_i^2 - \sum_{(i,j) \in \mathcal E}B_{ij}U_iU_j \cos(\theta_i - \theta_j) \nonumber \\
&+\frac 12 \sum_{i \in \mathcal V_\ell} \omega_{\ell,i}^2, \label{plant-hamiltonian}
\end{align}
which describes the total energy stored in the system. The first row of \eqref{plant-hamiltonian} represents the shifted kinetic energy of the rotors and the magnetic energy of the generator circuits, the second row represents the ``virtual'' kinetic energy at inverter nodes, the third row represents the magnetic energy of transmission lines and the fourth row represents the local deviations of load nodes from nominal frequency. 

Using the Hamiltonian $H_p(\bm x_p)$ and its gradient $\nabla H_p(\bm x_p)$, equations \eqref{powerflow1}--\eqref{planteq5} can be written as follows:
\begin{align}
\begin{bmatrix}
\dot{\bm \vartheta} \\ \dot{\bm L}_\mathcal G \\ \dot{\bm L}_\mathcal I \\ \dot{\bm U}_g \\ \bm 0 \\ \bm 0 \end{bmatrix}
&=\Vast[\underbrace{\begin{bmatrix}
	\bm 0 & \bm D_{p\mathcal G}^\top & \bm D_{p\mathcal I}^\top& \bm 0 &  \bm D_{p\ell}^\top & \bm 0 \\
	-\bm D_{p\mathcal G} & \bm 0 & \bm 0 & \bm 0 & \bm 0 & \bm 0  \\
	-\bm D_{p\mathcal I} & \bm 0 & \bm 0 & \bm 0 & \bm 0 & \bm 0  \\
	\bm 0 & \bm 0 & \bm 0 & \bm 0 & \bm 0 & \bm 0  \\
	-\bm D_{p\ell} & \bm 0 & \bm 0 & \bm 0 & \bm 0 & \bm 0 \\
	\bm 0 & \bm 0 & \bm 0 & \bm 0 & \bm 0 & \bm 0
	\end{bmatrix}}_{\bm J_p} \nonumber \\
&-\underbrace{\begin{bmatrix}
	\bm 0 & \bm 0 & \bm 0 & \bm 0 & \bm 0 & \bm 0 \\
	\bm 0 & \bm{A}_\mathcal G & \bm 0 & \bm 0 & \bm 0 & \bm 0 \\
	\bm 0 & \bm 0 & \bm{A}_\mathcal I & \bm 0 & \bm 0 & \bm 0 \\
	\bm 0 & \bm 0 &\bm 0& \bm{R}_g & \bm 0 & \bm 0 \\
	\bm 0 & \bm 0 & \bm 0 &\bm 0& \bm{A_\ell} & \bm 0 \\
	\bm 0 & \bm 0 & \bm 0 & \bm 0 &\bm 0& \widehat{\bm U}_\ell
	\end{bmatrix}}_{\bm R_p}\Vast]
\bm\nabla H_p \nonumber \\
&-
\underbrace{\begin{bmatrix}
	\bm 0 \\ \bm{\varphi}_\mathcal G \\ \bm{\varphi}_\mathcal I \\ \bm{\varrho}_\mathcal G \\ \bm{\varphi}_\ell \\ \bm{\varrho}_\ell
	\end{bmatrix}}_{\bm r_p}
+ 
	\begin{bmatrix}
	\bm 0 & \bm 0 & \bm 0 & \bm 0 & \bm 0 \\
	\bm {I} & \bm 0 &\bm 0 & \bm 0 & -\bm{\widehat I}_\mathcal G \\
	\bm 0 & \bm I &\bm 0 & \bm 0 & -\bm{\widehat I}_\mathcal I \\
	\bm 0 &\bm 0 & \bm{\hat \tau}_U & \bm 0 & \bm 0 \\
	\bm 0 & \bm 0 &\bm 0 & \bm 0 & -\bm{\widehat I}_\ell \\
	\bm 0 &\bm 0 & \bm 0 & -\bm I & \bm 0
	\end{bmatrix}
\begin{bmatrix}
\bm{p}_\mathcal G \\  \bm p_\mathcal I \\  \bm{U}_f \\ \bm{q}_\ell \\ \bm{p}_\ell  
\end{bmatrix}, \label{plantPHSlossy}
\end{align}
with
\begin{alignat}{5}
\bm A_\mathcal G &= \mathrm{ diag}_i\{A_i\},&&i \in \mathcal V_\mathcal G, \\
\bm A_\mathcal I &= \mathrm{ diag}_i\{A_i\},&&i \in \mathcal V_\mathcal I, \\
\bm A_\ell &= \mathrm{ diag}_i\{A_i\}, && i \in \mathcal V_\ell, \\
\bm R_\mathcal G &= \mathrm{ diag}_i\left\{\frac{X_{di}-X_{di}'}{\tau_{U,i}}\right\}, && i \in \mathcal V_\mathcal G, \\
\widehat{\bm U}_\ell &= \mathrm{ diag}_i\{U_i\}, && i \in \mathcal V_\ell,\\
\bm \varphi_\mathcal G &= \mathrm{ col}_i\big\{G_{ii}U_i^2 + \sum_{j \in \mathcal N_i}G_{ij}U_iU_j\cos(\vartheta_{ij})\big\}, && i \in \mathcal V_\mathcal G, \\
\bm \varphi_\mathcal I &= \mathrm{ col}_i\big\{G_{ii}U_i^2 + \sum_{j \in \mathcal N_i}G_{ij}U_iU_j\cos(\vartheta_{ij})\big\}, && i \in \mathcal V_\mathcal I, \\
\bm \varphi_\ell &= \mathrm{ col}_i\big\{G_{ii}U_i^2 + \sum_{j \in \mathcal N_i}G_{ij}U_iU_j\cos(\vartheta_{ij})\big\}, && i \in \mathcal V_\ell, \\
\bm \varrho_\mathcal G &= \mathrm{ col}_i\big\{R_{g,i}\sum_{j \in \mathcal N_i} G_{ij}U_iU_j\sin(\vartheta_{ij})\big\},&&i \in \mathcal V_\mathcal G, \\
\bm \varrho_\ell &= \mathrm{ col}_i\big\{\sum_{j \in \mathcal N_i} G_{ij}U_iU_j\sin(\vartheta_{ij})\big\},&&i \in \mathcal V_\ell, \\
\bm{\hat \tau}_U &= \mathrm{ diag}_i\{1 / \tau_{U,i}\}, && i \in \mathcal V_\mathcal G, \\
\widehat{\bm I}_\mathcal G &= \begin{bmatrix} \bm I_{n_g \times n_g} &  \bm 0_{n_g \times n_i} & \bm 0_{n_g \times n_\ell} \end{bmatrix},\\
\widehat{\bm I}_\mathcal I &= \begin{bmatrix} \bm 0_{n_i \times n_g} &  \bm I_{n_i \times n_i} & \bm 0_{n_i \times n_\ell} \end{bmatrix},\\
\widehat{\bm I}_\ell &= \begin{bmatrix} \bm 0_{n_\ell \times n_g} & \bm 0_{n_\ell \times n_i} &\bm I_{n_\ell \times n_\ell} \end{bmatrix}. 
\end{alignat}
Note that $\bm J_p = -\bm J_p^\top$ and $\bm R_p \succeq 0$. Hence, this is a port-Hamiltonian descriptor system \cite{Beattie2017} with a nonlinear dissipative relation due to $\bm r_p \neq \bm 0$ \cite{vanderSchaft.2017}.
\section{Price-Based Controller}\label{section:controller}
\subsection{Control Objective}
In the following controller design, the control variables $\bm p_{g}=\mathrm{col}\{\bm p_\mathcal G, \bm p_\mathcal I\}$ are to be regulated in such a way that the steady-state frequency deviation $\overline{\omega}$ from the nominal frequency $\omega^{n}$ is zero at each node while, at the same time, the steady-state generation $\overline{\bm p}_g$ is optimal with respect to an objective function being to be defined. 

In \cite{Trip.2016}, it is shown that a necessary condition for $\bm \omega = \bm 0$ is that the overall resistive losses
\begin{align}
\Phi = \sum_{i \in \mathcal V} G_{ii}U_i^2 + 2 \cdot \sum_{(i,j) \in \mathcal E} G_{ij} U_iU_j\cos(\vartheta_{ij}) \label{resistive}
\end{align}
are equal to the net sum of generation and load, i.e.
\begin{align}
\Phi\stackrel !=\sum_{i \in \mathcal V_\mathcal G}{p_{\mathcal G,i}}+  \sum_{i \in \mathcal V_\mathcal I}{p_{\mathcal I,i}} - \sum_{i \in \mathcal V}{p_{\ell,i}}. \label{eq-balance-scalar}
\end{align}
The above condition serves as a fundamental constraint for any equilibrium that the closed-loop system is supposed to attain. As shown e.g. in \cite{Stegink.2017}, for a given $\bm p_\ell$, the allocation of active power injections $\bm p_g$ is a solution of \eqref{eq-balance-scalar} if and only if there exists a $\bm \nu \in \mathds R^{m_c}$ such that
\begin{align}
\bm D_c \bm \nu = \widehat{\bm I}_\mathcal G^\top\bm p_\mathcal G  +  \widehat{\bm I}_\mathcal I^\top\bm p_\mathcal I - \bm p_{\ell} - \bm\varphi, \label{eq-balance-komm}
\end{align}
with $\bm D_c$ being an arbitrary incidence matrix of a communication graph $\mathscr G_c=(\mathcal V,\mathcal E_c)$ with $m_c= |\mathcal E_c|$ edges and $\bm\varphi = \mathrm{col}\{\bm\varphi_\mathcal G , \bm \varphi_\mathcal I ,  \bm\varphi_\ell\}$.

This alternative formulation by means of \eqref{eq-balance-scalar} will result in a \emph{distributed} controller with control variables $p_{g,i}$ being only dependent on variables of node $i$ or on variables that are adjacent with respect to the communication graph.

The aim is now to design a controller in such a way that the closed-loop equilibrium, i.e. the steady state, is a solution to the optimization problem
\begin{align*}
\begin{array}{ll}
\displaystyle \min_{\bm p_{\mathcal G}, \bm p_{\mathcal I}, \bm \nu}&C(\bm p_{\mathcal G}, \bm p_{\mathcal I}) \\
\mathrm{subject \ to} &\eqref{eq-balance-komm}
\end{array}
\tag{OP}\label{opt-P1}
\end{align*} 
with $C(\bm p_{\mathcal G}, \bm p_\mathcal I)$ being an arbitrary, strictly convex cost function.
% e.g. the sum of all generation costs at each node.

Note that \eqref{opt-P1} is a convex optimization problem since \eqref{eq-balance-komm} is linear-affine with respect to optimization variables $(\bm p_{\mathcal G}, \bm p_{\mathcal I}, \bm \nu)$.

\subsection{Distributed Control Algorithm}
The primal-dual gradient method for convex optimization problems \cite{Jokic.2007,JokicLazarvandenBosch2009,Trip.2016} is used to derive a controller that solves \eqref{opt-P1} in steady state.

To simplify the notation we define
\begin{align}
\widehat{\bm I}_g = \begin{bmatrix} \widehat{\bm I}_\mathcal G \\ \widehat{\bm I}_\mathcal I \end{bmatrix}
\end{align}
and by letting \cite[Proposition 1]{Koelsch2019} apply, we get the distributed controller
\begin{align}
\bm \tau_g \dot{\bm p}_{g} &= - \nabla C({\bm p}_{g})+ \widehat{\bm I}_g{\bm \lambda} + \bm u_c, \label{primal-dual-1}\\
\bm \tau_{\lambda}\dot{\bm \lambda}&= \bm D_{c} \bm \nu - \widehat{\bm I}_g^\top\bm p_{g}+  \bm p_{\ell} + \bm \varphi, \label{primal-dual-2} \\
%\bm \tau_{\lambda_\ell}\dot{\bm \lambda}_\ell &= \bm D_{c\ell} \bm \nu + \widehat{\bm I}_\ell  \bm p_{\ell} + \bm \varphi_\ell \\
\bm \tau_\nu \dot{\bm \nu} &= -\bm D_{c}^\top \bm \lambda. \label{primal-dual-4}
\end{align}
Diagonal matrices $\bm \tau_g, \bm \tau_\lambda, \bm \tau_\nu > 0$ are used to adjust the convergence behavior of the respective variable: The smaller the $\tau$ value, the faster the convergence and the larger the transient amplitudes. $\bm u_c$ is an additional controller input which is later chosen in such a way that a power-preserving interconnection of plant and controller is achieved \cite{vanderSchaft.2017}, resulting in a closed-loop system which is again port-Hamiltonian.

By defining the controller state 
$\bm x_c = \mathrm{col}\{ \bm \tau_g \bm p_g , \bm \tau_{\lambda}\bm \lambda,  \bm \tau_{\nu}\bm{\nu}\}$
and the controller Hamiltonian
\begin{align}
H_c(\bm x_c)=\frac 12 \bm x_c^\top \bm \tau_c^{-1}\bm x_c
\end{align}
with
\begin{align}
\bm \tau_c=\mathrm{diag}\{\bm \tau_g, \bm \tau_{\lambda},  \bm \tau_{\nu}\},
\end{align}
controller equations \eqref{primal-dual-1}--\eqref{primal-dual-4} have the port-Hamiltonian representation
\begin{align}
\dot{\bm x}_c = \underbrace{\begin{bmatrix} \bm 0 & \widehat{\bm I}_g & \bm 0  \\
	-\widehat{\bm I}_g^\top & \bm 0 & \bm D_{c} \\
	\bm 0 & -\bm D_{c}^\top& \bm 0 \end{bmatrix}}_{\bm J_c} \nabla H_c - \underbrace{\begin{bmatrix} \nabla C \\ -\bm \varphi \\  \bm 0 \end{bmatrix}}_{\bm r_c} 	+ \begin{bmatrix} \bm u_c \\  \bm p_{\ell} \\ \bm 0
\end{bmatrix}. \label{eq-controller-phs}
\end{align}
This representation now provides a straightforward way to set up and analyze the closed-loop system.
\subsection{Closed-Loop System}
With the new composite Hamiltonian $H(\bm x_p, \bm x_c) = H_p(\bm x_p) + H_c(\bm x_c)$ and by choosing $\bm u_c = -\mathrm{col}\{\bm \omega_\mathcal G, \bm \omega_\mathcal I\}$ as control input \cite{Stegink.2017,Stegink.,Koelsch2019}, the interconnection of plant and controller results in the closed-loop descriptor system 
\begin{align}
\bm E\dot{\bm x} = \left(\bm J - \bm R  \right) \nabla H - \bm r + \bm F \bm u \label{closed-loop}
\end{align}
with
\begin{alignat}{2}
\bm E &= \mathrm{diag}\{ \bm I_{3n_\mathcal G +2n_\mathcal I  + n+m+m_c}, \bm 0_{2n_\ell}\},\\
\bm J &= 
\begin{bmatrix}
\bm 0 & \widehat{\bm I}_g & \bm 0& \bm 0 & -\widetilde{\bm I}_\mathcal G^\top &  -\widetilde{\bm I}_\mathcal I^\top & \bm 0 & \bm 0 & \bm 0  \\
-\widehat{\bm I}_g^\top & \bm 0 & \bm D_{c} & \bm 0 & \bm 0& \bm 0 & \bm 0 & \bm 0 & \bm 0  \\
\bm 0 & -\bm D_{c}^{\top}   & \bm 0 & \bm 0 & \bm 0 & \bm 0  & \bm 0 & \bm 0& \bm 0 \\
\bm 0 & \bm 0 & \bm 0 & \bm 0 & \bm D_{p\mathcal G}^\top & \bm D_{p\mathcal I}^\top & \bm 0 & \bm D_{p\ell}^\top & \bm 0 \\
\widetilde{\bm I}_\mathcal G & \bm 0 & \bm 0 &  -\bm D_{p\mathcal G} & \bm 0 & \bm 0 & \bm 0 & \bm 0 & \bm 0  \\
\widetilde{\bm I}_\mathcal I  & \bm 0 & \bm 0  &  -\bm D_{p\mathcal I} & \bm 0 & \bm 0& \bm 0 & \bm 0  & \bm 0  \\
\bm 0 & \bm 0 & \bm 0 &  \bm 0 & \bm 0 & \bm 0 & \bm 0 & \bm 0 & \bm 0  \\
\bm 0 & \bm 0 & \bm 0 &  -\bm D_{p\ell} & \bm 0 & \bm 0 & \bm 0 & \bm 0 & \bm 0 \\
\bm 0 & \bm 0 & \bm 0 & \bm 0 & \bm 0 & \bm 0 & \bm 0 & \bm 0 & \bm 0  
\end{bmatrix}, \\
\bm R &= \mathrm{diag}\{\bm 0_{n_\mathcal G + n_\mathcal I +n+m_c}, \bm R_p \}, \\
\bm r &= \mathrm{col}\{\bm r_c, \bm r_p\}, \\
\bm F &= \begin{bmatrix} \bm 0 & \bm 0 & \bm 0 \\
\bm 0 & \bm 0 & {\bm I} \\
\bm 0 & \bm 0 & \bm 0 \\
\bm 0 & \bm 0 & \bm 0 \\
\bm 0 & \bm 0 & -\bm{\widehat I}_\mathcal G \\
\bm 0 & \bm 0 & -\bm{\widehat I}_\mathcal I \\
\bm{\hat \tau}_U & \bm 0 & \bm 0 \\
\bm 0 & \bm 0 & -\bm{\widehat I}_\ell \\
\bm 0 & -\bm I & \bm 0
\end{bmatrix}, \\
\bm u &= \mathrm{col}\{\bm U_f, \bm q_\ell , \bm p_\ell \} \\
\widetilde{\bm I}_\mathcal G &= \begin{bmatrix} \bm I_{n_\mathcal G \times n_\mathcal G} & \bm 0_{n_\mathcal G \times n_\mathcal I}\end{bmatrix} \\
\widetilde{\bm I}_\mathcal I &= \begin{bmatrix} \bm I_{n_\mathcal I \times n_\mathcal G} & \bm 0_{n_\mathcal I \times n_\mathcal I}\end{bmatrix}.
\end{alignat}
Due to $\bm J = -\bm J^\top$ and $\bm R \succeq 0$ the system is again port-Hamiltonian.

Denote each equilibrium $\overline{\bm x}$ of \eqref{closed-loop} by $\overline{\bm x}$. This equilibrium has two salient properties which are presented in the following two propositions: 
\begin{proposition}
	At each equilibrium of \eqref{closed-loop}, the frequency deviation $\overline \omega_i$, $i \in \mathcal V$, from nominal frequency $\omega^{n}$ is zero.
\end{proposition}
\begin{proof}
	Let $\mathds 1^\top= \begin{bmatrix} 1 & \cdots & 1 \end{bmatrix}$ be the all-ones row vector. 	With $\dot{\overline{\bm \vartheta}}= \bm 0$ at steady state, the first row of \eqref{plantPHSlossy} equals $\bm 0 = \bm D_p^\top \overline{\bm \omega}$. Since $\bm D_p$ is the incidence matrix of a connected graph, this implies that each row of vector $\overline{\bm \omega}$ has the same value, i.e. $\overline{\bm \omega} = \overline \omega \cdot \mathds 1$, and thus each node of the microgrid is synchronized to a common frequency $\overline \omega$.  
		
	Since $\mathds 1^\top \bm \varphi = \Phi$ and $\dot{\overline{\bm \lambda}} = \bm 0$,
	left-multiplying \eqref{primal-dual-2} with $\mathds 1^\top$ yields 
	\begin{align}
	0 = 0 -	\sum_{i \in \mathcal V_\mathcal G \cup \mathcal V_\mathcal I}{\overline p_{g,i}} + \sum_{i \in \mathcal V}{\overline p_{\ell,i}} + \overline \Phi, 
	\end{align}
	i.e. condition \eqref{eq-balance-scalar} is fulfilled at steady state. Moreover, a comparison of \eqref{powerflow1} and \eqref{resistive} shows that $\sum_{i \in \mathcal V}p_i = \Phi$. Now left-multiplying \eqref{planteq2} with $\mathds 1^\top$ equals 
\begin{align}
0 = - \sum_{i \in \mathcal V_{\mathcal G}\cup \mathcal V_\mathcal I} A_i \overline \omega_i= - \overline \omega \cdot \sum_{i \in \mathcal V_{\mathcal G}\cup \mathcal V_\mathcal I} A_i
\end{align}
 and since $A_i >0$, it follows that $\overline \omega$ must be zero.
\end{proof}
\begin{proposition}
	At each equilibrium $\overline{\bm x}$ of \eqref{closed-loop}, the marginal prices are equal, i.e. $\nabla C(\overline p_{g,i}) = \nabla C(\overline p_{g,j})$ for $i,j \in \mathcal V_{\mathcal G} \cup \mathcal V_{\mathcal I}$.
\end{proposition}
\begin{proof}
	With $\dot{\overline{\bm \nu}}= \bm 0$ at steady state, \eqref{primal-dual-4} equals $\bm 0 = - \bm D_c^\top \overline{\bm \lambda}$. Since $\bm D_c$ is the incidence matrix of a connected graph, this implies that each row of $\overline{\bm \lambda}$ has the same value, i.e. $\overline{\bm \lambda} = \overline \lambda \cdot \mathds 1$,
	
	Moreover, with $\overline{\bm \omega} = \bm 0 $ from Proposition 1, \eqref{primal-dual-1} leads to $\overline{\bm \lambda}= \nabla C(\overline{\bm p}_g)$ at steady state and hence all marginal prices are equal.
\end{proof}

Proposition 2 shows that the closed-loop system fulfills the well-known economic dispatch criterion \cite{Doerfler2019} at steady state. Note that in this context, $\bm \lambda$ can be interpreted as a \emph{price signal}.

The stability of the closed-loop equilibrium can be investigated by exploiting the port-Hamiltonian structure  \eqref{closed-loop} with its (shifted) passivity property: With dissipation vector $\mathcal R(\bm x)= \bm R \nabla H (\bm x) + \bm r$, equation \eqref{closed-loop} reads as follows:
\begin{align}
\bm E\dot{\bm x} = \bm J \nabla H(\bm x) - \mathcal R(\bm x) + \bm F \bm u, \label{closed-phs-1}
\end{align}
with each equilibrium $\overline{\bm x}$ fulfilling 
\begin{align}
{\bm 0} = \bm J \nabla H(\overline{\bm x}) - \mathcal R(\overline{\bm x}) + \bm F \overline{\bm u} \label{closed-phs-2}
\end{align}
for a constant input vector $\overline{\bm u}$. Since $H(\bm x)$ is a convex and nonnegative function, the shifted Hamiltonian \cite{vanderSchaft.2017}
\begin{align}
\overline H(\bm x):=H(\bm x)-\left( \bm x - \overline{\bm x} \right)^\top \nabla H(\overline{\bm x}) - H(\overline{\bm x}) \label{shifted-H}
\end{align}
is positive definite with minimum $\overline H(\overline{\bm x})=0$. Thus the shifted closed-loop dynamics, i.e. \eqref{closed-phs-1} minus \eqref{closed-phs-2}, can be expressed in terms of $\overline H(\bm x)$  as follows:
\begin{align}
\bm E\dot{\bm x} = \bm J \nabla \overline H(\bm x) - \left[\mathcal R(\bm x)- \mathcal R(\overline{\bm x})\right] + \bm F \left[\bm u - \overline{\bm u}\right].
\end{align}
As a result, stability of $\overline{\bm x}$ is given if the shifted passivity property \cite{vanderSchaft.2017}
\begin{align}
\left[\nabla H(\bm x) - \nabla H({\overline{\bm x}})\right]^\top \left[\mathcal R(\bm x)- \mathcal R(\overline{\bm x})\right] \geq 0 \label{eq-passivity}
\end{align}
is satisfied. Note that for $\bm r_p = \bm 0$, i.e. lossless microgrids, \eqref{eq-passivity} is always fulfilled due to strict convexity of $C(\bm p_\mathcal G, \bm p_\mathcal I)$.  
%\lukas{evtl. unverstaendlich}
\section{Simulation}\label{section:simulation}
\subsection{Case Study}
The price-based steady state optimal controller presented in the previous section is now demonstrated by means of an 18-node exemplary microgrid with base voltage of \SI{10}{\kilo\volt}, $n_\mathcal G=n_\mathcal I=7$ and $n_\ell = 4$, see Fig. \ref{fig:grid}.  Generator nodes are represented by black nodes, inverter nodes are represented by gray nodes, and load nodes are represented by white nodes.
 \begin{figure}[t]
	\centering
	\includegraphics[clip, trim=0cm 4cm 0cm 4cm, width=\columnwidth]{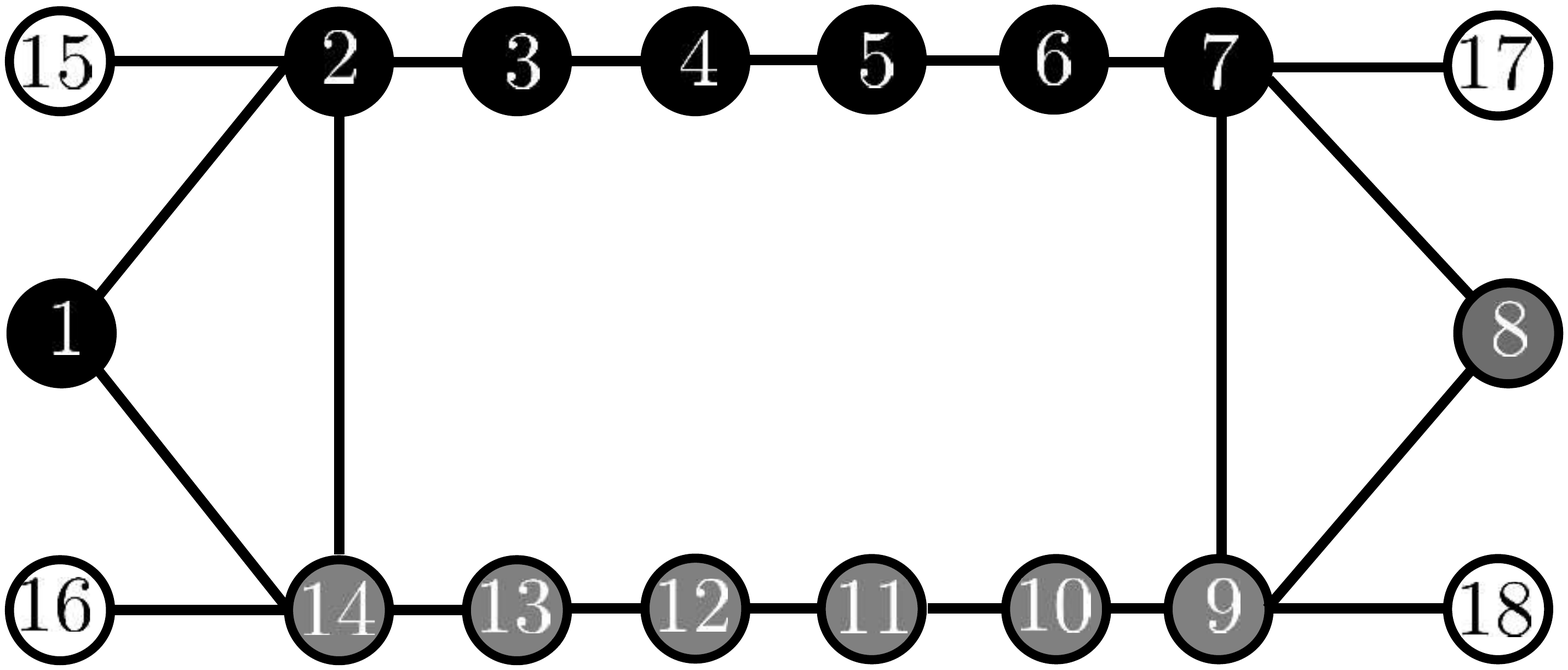}
	\caption{Network topology of exemplary microgrid.}	
	\label{fig:grid}
\end{figure}
All parameters of the microgrid can be found in Tables \ref{tab:generator} to \ref{tab:line} and are given in p.u., except $\tau_{U,i}$, which is given in seconds.
The numericals values for the parameters of generator nodes, load nodes and transmission lines are based on those provided in \cite{Trip.2016,Koelsch2019} and the parameter values of inverter nodes base upon \cite{Monshizadeh.}.
\begin{table}
	\caption{Parameters of Generator Nodes}
	\begin{center}
		\begin{tabular}{|l||l|l|l|l|l|l|l|l|}
			\hline
			$i$ 	& 1 & 2 & 3 & 4 & 5 & 6 & 7  \\
			\hline
			\hline
			$A_i$ & 1.6 & 1.22 & 1.38 & 1.42 & 1.4 & 1.3 & 1.3 \\
			$B_{ii}$ & -2.67 & -6.97 & -4.0 & -2.1 & -3.5 & -5.5 & -7.2\\
			$M_i$ & 5.2 & 3.98 & 4.49 & 4.22 & 4.4 & 4.5 & 5.15  \\ 
			$X_{d,i}$ & 0.02 & 0.03 & 0.03 & 0.025 & 0.02 & 0.024 & 0.03\\
			$X'_{d,i}$ & 0.004 & 0.006 & 0.005 & 0.005 & 0.003 & 0.0044 & 0.0068 \\
			$\tau_{U,i}$ & 6.45 & 7.68 & 7.5 & 6.5 & 6.9 & 7.2 & 6.88 \\
			\hline
		\end{tabular}
		\label{tab:generator}
	\end{center}
\end{table}
\begin{table}
	\caption{Parameters of Inverter Nodes}
	\begin{center}
		\begin{tabular}{|l||l|l|l|l|l|l|l|}
			\hline
			$i$ & 8 & 9 & 10 & 11 & 12 & 13 & 14 \\
			\hline
			\hline
			$A_i$ & 1.5 & 1.7 & 1.55 & 1.6 & 1.4 & 1.65 & 1.25 \\
			$B_{ii}$ & -6.2 & -7.1 & -4.5 & -4.2 & -4.5 & -6.05 & -7.1 \\
			$M_i$ & 4 & 3.85 & 6 & 5.55 & 4.1 & 3.9 & 4.32   \\ 
			\hline
		\end{tabular}
		\label{tab:inverter}
	\end{center}
\end{table}
\begin{table}
	\caption{Parameters of Load Nodes}
	\begin{center}
		\begin{tabular}{|l||l|l|l|l|}
			\hline
			$i$  & 15 & 16 & 17 & 18 \\
			\hline
			\hline
			$A_i$  & 1.45 & 1.35 & 1.5 & 1.7 \\
			$B_{ii}$ & -2.05 & -2.2 & -1.5 & -2.1 \\
			\hline
		\end{tabular}
		\label{tab:load}
	\end{center}
\end{table}

\begin{table}[t]
	\centering
		\caption{Parameters of Transmission Lines}
	\begin{tabular}{|l||l|}
		\hline
		$B_{1,2}$ & 1.27 \\
		$B_{1,14}$ & 1.4 \\
		$B_{2,3}$ & 1.4 \\
		$B_{2,14}$ & 2.25 \\
		$B_{2,15}$ & 2.05 \\
		$B_{3,4}$ & 1.1 \\
		$B_{4,5}$ & 1.0 \\
		\hline
		\end{tabular}
	\hspace{0.2cm}
	\begin{tabular}{|l||l|}
	\hline
	$B_{5,6}$ & 2.5 \\
	$B_{6,7}$ & 3.0 \\
	$B_{7,8}$ & 2.7 \\
	$B_{7,9}$ & 1.5 \\
	$B_{7,17}$ & 3.0 \\
	$B_{8,9}$ & 3.5 \\
	$B_{9,10}$ & 1.5 \\
	\hline
\end{tabular}
	\hspace{0.2cm}
\begin{tabular}{|l||l|}
	\hline
	$B_{9,18}$ & 2.1 \\
	$B_{10,11}$ & 3.0 \\
	$B_{11,12}$ & 1.2 \\
	$B_{12,13}$ & 3.3 \\
	$B_{13,14}$ & 1.25 \\
	$B_{14,15}$ & 2.2 \\
	\hline
	\multicolumn{2}{c}{} 
\end{tabular}
\label{tab:line}
\end{table}
%\lukas{hier müsste man eigtl. noch klären, inwiefern diese Wahl der Parameter gerechtfertigt ist}.
 Without loss of generality, yet for sake of simplicity, we assume constant $R / X$ ratios $\gamma$, i.e. $G_{ij} = -\gamma \cdot B_{ij}$ for each line $(i,j)$, and $\bm \tau_c =\num{0.01} \cdot \bm I$. Moreover, we choose $\bm D_c$ to be identical to the plant incidence matrix $\bm D_p$ after it has been pointed out in \cite{Koelsch2019} that the specific choice of $\bm D_c$ has little influence on the convergence speed to the desired equilibrium.

The simulations were carried out in Wolfram Mathematica 12.0.

\subsection{Cost function and input signals}
The cost function is chosen to 
\begin{align}
C(\bm p_g)=\frac 12 \sum_{i \in \mathcal V_\mathcal G \cup V_\mathcal I} \frac{1}{w_i} \cdot p_{g,i}^2,
\end{align}
with weighting factors $\omega_1 = 1, w_2=\num{1.1}, w_{3}=\num{1.2}$ and so on. Bearing in mind Proposition 2, this specific choice of $C(\bm p_g)$ as a weighted sum of squares leads to active power sharing \cite{Schiffer.1} in steady state, i.e. a proportional share $\overline p_{g,i} \slash w_i  = \overline p_{g,j} \slash w_j= \text{const.}$ for all $i,j\in \mathcal V_\mathcal G \cup \mathcal V_\mathcal I$.

The initial values of input vector $\bm u$ and state vector $\bm x$ are chosen such that the closed-loop system starts in synchronous mode with $\bm \omega ( t=0) = \bm 0$. At regular intervals of \SI{100}{\second}, a step of +\num{0.5} p.u. occurs sequentially at each load node, as shown in Fig. \ref{fig:pl}. $\gamma$ is set to one.
 \begin{figure}[t]
	\centering
	\includegraphics[width=\columnwidth]{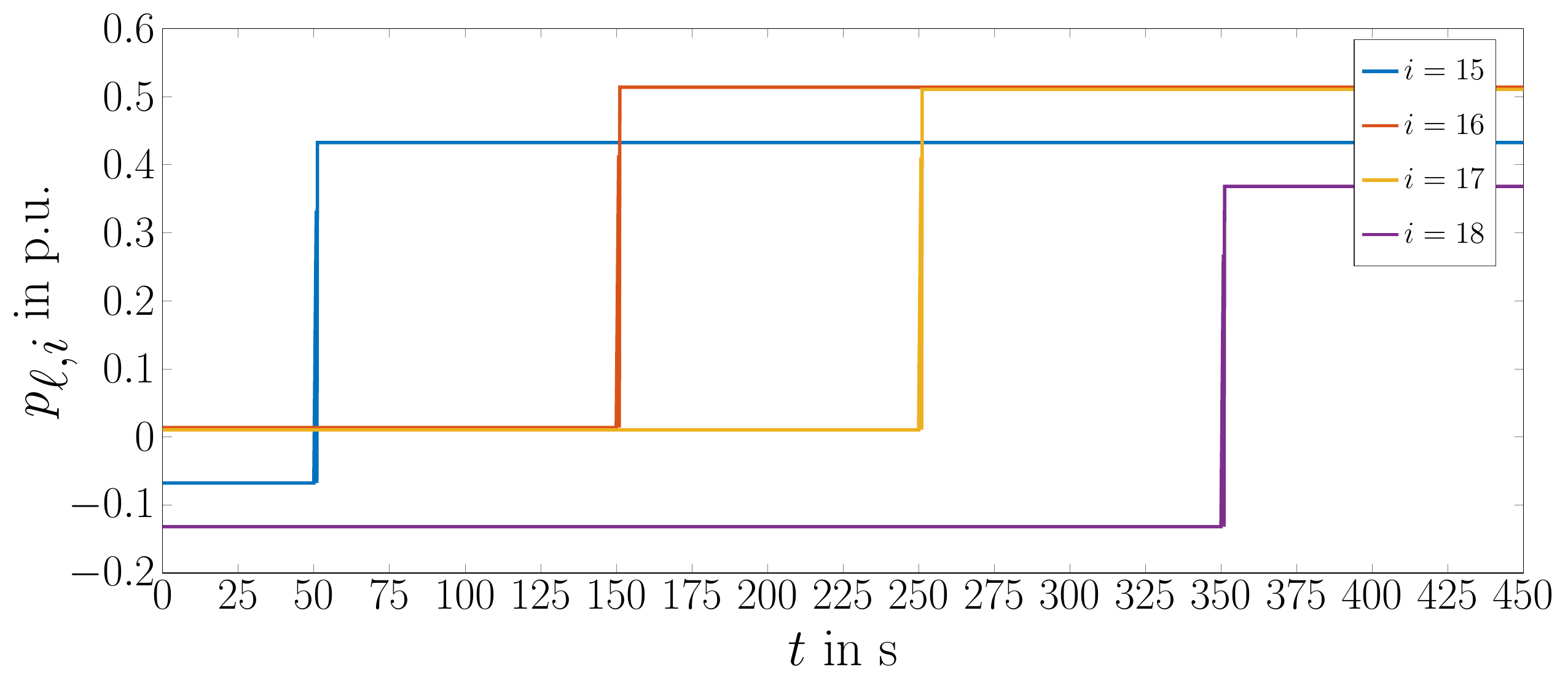}
	\caption{Stepwise increase at load nodes.}	
	\label{fig:pl}
\end{figure}
 \begin{figure*}[t!]
	\centering
	\includegraphics[width=\textwidth]{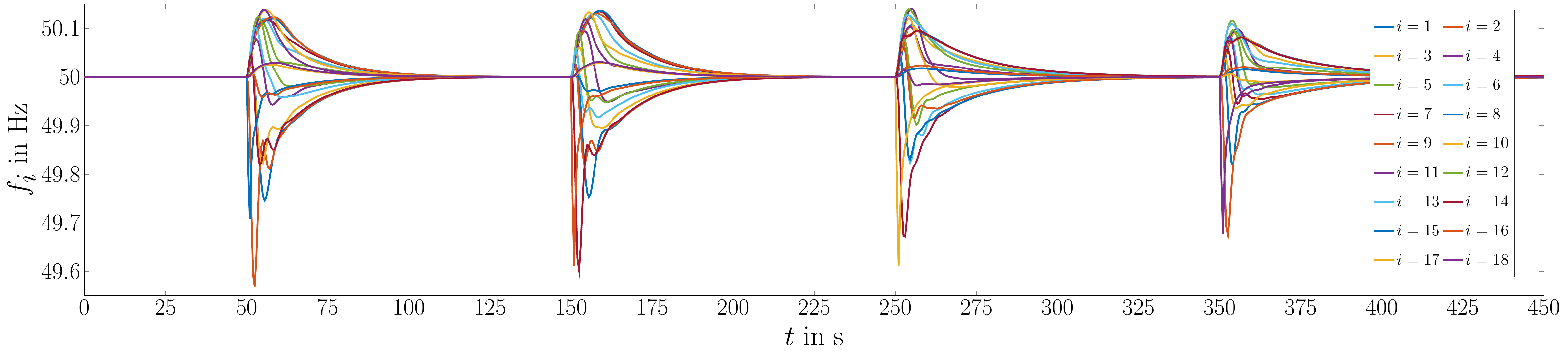}
	\caption{Frequency regulation after step increase at load nodes.}	
	\label{fig:f}
\end{figure*}
\subsection{Results}
Fig. \ref{fig:f} shows the node frequencies for each $i \in \mathcal V$. Starting from synchronous mode with a frequency of \SI{50}{Hz.} at each node, a deviation of the local frequency in the range of about \SI{-0.45}{Hz.} to $+$\SI{0.1}{Hz.}  occurs immediately after the load jumps, before being resynchronized again and being regulated to \SI{50}{Hz.} by the controller. The convergence speed of the individual frequencies to the common frequency of \SI{50}{Hz.} is independent of which node the load jump occurred at.

%%%%%%%%%%%%%%%%%FIG. 3: Frequenzverläufe%%%%%%%%%%%%%%%%%%%%

Fig. \ref{fig:pinj} shows the corresponding active power generation $\bm p_g$ at generator and inverter nodes. After each load step, the controllers automatically increase $\bm p_g$ to compensate for the additional demand. Remarkably, the individual power injections $p_{g,i}$ are equidistant from each other at steady state, regardless of the total generation, thus active power sharing is evident. 

The decay time of both frequency deviation and power regulation is about \SI{40}{\second}. It can be further accelerated by choosing smaller entries within matrix $\bm \tau_c$ of the controller.
 \begin{figure}[t]
	\centering
	\includegraphics[width=\columnwidth]{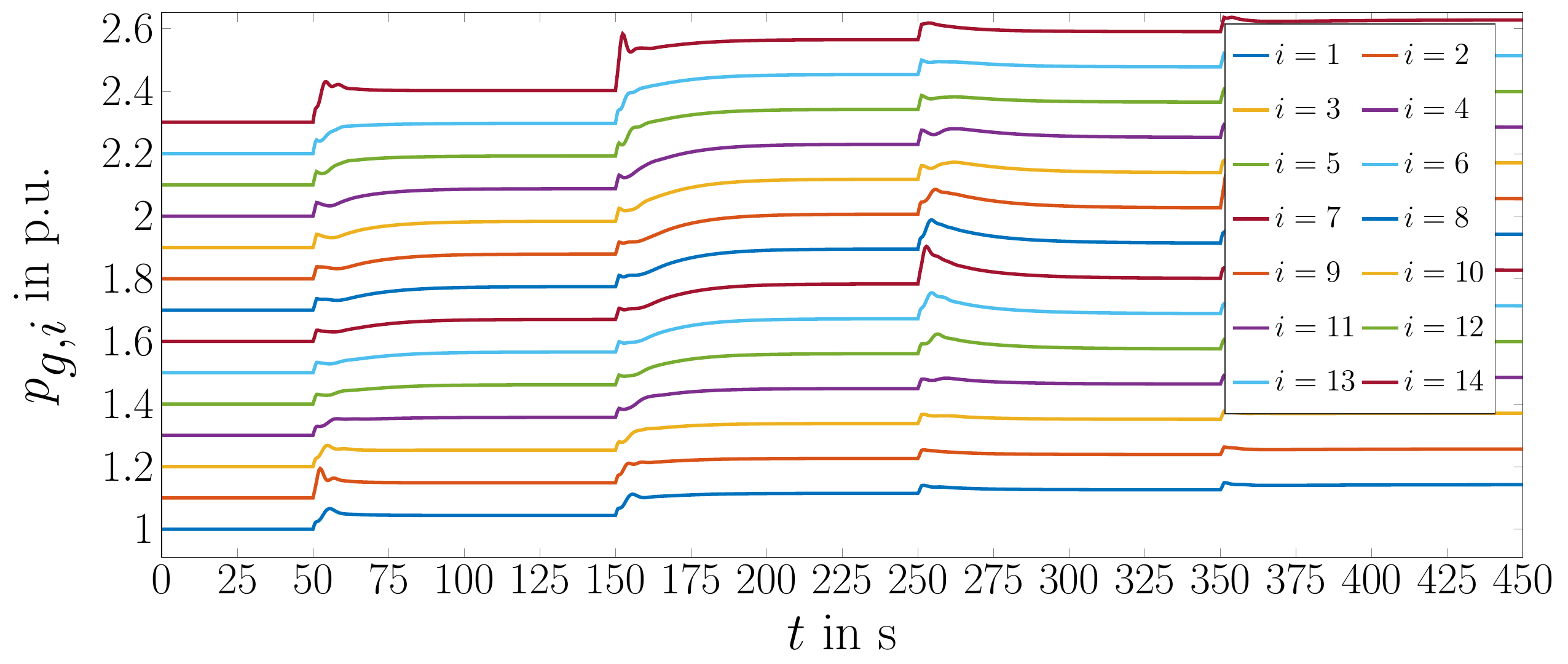}
	\caption{Active power generation at generator and inverter nodes.}	
	\label{fig:pinj}
\end{figure}
%%%%%%%%%%%%%%%%%FIG. 5: Dissipations-Funktion%%%%%%%%%%%%%%%%%%%%
\section{Conclusion and Future Work}\label{section:conclusion} 
In this paper we presented a model-based, steady state optimal controller for heterogeneous microgrids. The underlying microgrid model for the controller can consist of a mixture of conventional synchronous generators, power electronics interfaced sources and uncontrollable loads. In contrast to state-of-the-art approaches, the controller ensures asymptotic stability of the equilibrium at nominal frequency of \SI{50}{Hz} even with nonzero line resistances. 
The controller also provides an automatic solution to an optimization problem with a user-definable cost function. As shown in a simulation example, active power sharing can thus be achieved, for instance. However, other optimization problems can also be addressed, e.g. minimal total power input or minimal total line losses.
The closed-loop dynamics can be formulated as a port-Hamiltonian system and thus asymptotic stability of the overall system can be shown using a (shifted) passivity property.

In future research, an integrated voltage regulation will be incorporated into the existing controller scheme.
Furthermore, the presented controller will be applied to a benchmark system with significantly larger amount of nodes and real-world generation and load profiles to further illustrate the feasibility for large-scale systems.
\bibliography{Stegink_GIL_Quellen} 
\bibliographystyle{unsrt}
\end{document}